\documentclass[12pt]{article}
\usepackage[margin=2.5cm]{geometry}

\usepackage{mathptmx}       
\usepackage{helvet}         
\usepackage{courier}        
\usepackage{type1cm}        
\usepackage{makeidx}         
\usepackage{graphicx}        
\usepackage{multicol}        
\usepackage[bottom]{footmisc}

\usepackage{tikz,pgfplots,pstricks,pst-node}
\usepackage{color,array,colortbl}
\usepackage{amsmath,amsthm,amsfonts,bm}
\usepackage{url}

\newcolumntype{A}{>{\columncolor[gray]{0.8}}c}
\newcolumntype{B}{>{\columncolor[gray]{0.6}}c}

\newcommand{\bsu}{\boldsymbol{u}}    
\newcommand{\bsv}{\boldsymbol{v}}    
\newcommand{\bsz}{\boldsymbol{z}}    
\newcommand{\bsZ}{\boldsymbol{Z}}    

\DeclareMathOperator{\Var}{Var}
\newcommand{\Prob}{\mathbb{P}}
\newcommand{\Exp}{\mathbb{E}}

\def\citep#1#2{\cite[{#1}]{#2}}

\theoremstyle{plain}
  \newtheorem{theorem}{Theorem}
  \newtheorem{lemma}{Lemma}
  
  \newtheorem{proposition}{Proposition}
\theoremstyle{definition}

\theoremstyle{remark}

\newcommand{\RefEq}[1]{~\textup{(\ref{#1})}}
\newcommand{\RefEqEq}[2]{~\textup{(\ref{#1})} and~\textup{(\ref{#2})}}

\newcommand{\RefSec}[1]{Section~\textup{\ref{#1}}}
\newcommand{\RefSecSec}[2]{Sections~\textup{\ref{#1}} and~\textup{\ref{#2}}}

\newcommand{\RefProp}[1]{Proposition~\textup{\ref{#1}}}
\newcommand{\RefLem}[1]{Lemma~\textup{\ref{#1}}}

\newcommand{\RefFig}[1]{Figure~\textup{\ref{#1}}}

\newcommand{\RefTab}[1]{Table~\textup{\ref{#1}}}

\begin{document}

\title{Conditional sampling for barrier option pricing under the Heston model}
\author{Nico Achtsis \and Ronald Cools \and Dirk Nuyens}
\maketitle


\abstract{
We propose a quasi-Monte Carlo algorithm for pricing knock-out and knock-in barrier options under the Heston (1993) stochastic volatility model.
This is done by modifying the LT method from Imai and Tan (2006) for the Heston model such that the first uniform variable does not influence the stochastic volatility path and then conditionally modifying its marginals to fulfill the barrier condition(s).
We show that this method is unbiased and never does worse than the unconditional algorithm.
In addition, the conditioning is combined with a root finding method to also force positive payouts.
The effectiveness of this method is shown by extensive numerical results.
}


\section{Introduction}

It is well known that the quasi-Monte Carlo method in combination with a good path construction method, like the LT method from Imai and Tan \cite{IT2006}, can be a helpful tool in option pricing, see, e.g., \cite{GKSW2008,LEC2009}.
The integrand functions usually take the form $\max(f,0)$ and a good path construction will somehow align the discontinuity in the derivative along the axes.
However, as soon as other discontinuities, in the form of barrier conditions, are introduced, the performance of the quasi-Monte Carlo method degrades, see \cite{NW2012} for an illustration and an alternative solution.
This is also the case for the Monte Carlo method for which in  \cite{GS2001} a conditional sampling method has been introduced to alleviate this problem.

In previous work \cite{NicoBSLT} we have introduced a conditional sampling method to deal with barrier conditions in the Black--Scholes setting that can be used in combination with a good path construction method like the LT method.
In that paper we have shown that such a scheme always performs better than the unconditional method.
Here we consider the more realistic Heston model \cite{Heston1993}, which has a stochastic volatility component, and derive an algorithm to do conditional sampling on barrier conditions under this model.
We focus solely on the LT path construction which enables us to construct a good path construction for the payoff; excluding the maximum and barrier conditions which are handled by a root finding method (optional) and the conditional sampling proposed in this paper.


\section{The LT method for Heston under log prices}\label{sec:LTmethod}

Assume a Heston world \cite{Heston1993} in which the risk-neutral dynamics of the asset are given by
\begin{align*}
  dS(t)
  &=
  rS(t)dt + \sqrt{V(t)}S(t)dW^1(t)
  , \\
  dV(t)
  &=
  (\theta-V(t))\kappa dt+\sigma\sqrt{V(t)}dW^2(t)
  , \\
  dW^1(t)dW^2(t)
  &=
  \rho dt
  ,
\end{align*}
where $S(t)$ denotes the price of the asset at time $t$, $r$ is the risk-free interest rate, $\kappa$ is the mean-reversion parameter of the volatility process, $\theta$ is the long run average price variance and $\sigma$ is the volatility of the volatility. 
We assume the Feller condition $2 \kappa \theta \ge \sigma^2$ such that the process $V(t)$ is strictly positive.
The parameter $\rho$ controls the correlation between the log-returns and the volatility.
A useful observation is that one can write
\begin{align*}
 W^1(t)
 &=
 \rho W^2(t) + \sqrt{1-\rho^2} \, W^3(t),
\end{align*}
where $W^2(t)$ and $W^3(t)$ are independent Brownian motions.
This corresponds to the Cholesky decomposition of the correlation structure.
When resorting to Monte Carlo techniques for pricing options under this model, asset paths need to be discretized. 
For simplicity we assume that time is discretized using $m$ equidistant time steps $\Delta t=T/m$, but all results can be extended to the more general case.
The notations $\hat{S}_k$ and $\hat{V}_k$ will be used for $\hat{S}(k\Delta t)$ and $\hat{V}(k\Delta t)$, respectively.
We use the Euler--Maruyama scheme \cite{KP1992} to discretize the asset paths in log-space (see also \cite[Sect.\ 6.5]{Glass2003} w.r.t.\ transformations of variables) and sample the independent Brownian motions $W^2$ and $W^3$ by using independent standard normal variables $Z^1$ and $Z^2$; for $k=0,\ldots,m-1$,
\begin{align}
  \label{eq:logS}
  \log\hat{S}_{k+1}
  &=
  \log\hat{S}_{k}+\left(r-\frac{\hat{V}_k}{2}\right)\Delta t + \sqrt{\hat{V}_k}\sqrt{\Delta t}\left(\rho Z^1_{k+1}+\sqrt{1-\rho^2}Z^2_{k+1}\right)
  , \\
  \label{eq:V}
  \hat{V}_{k+1}
  &=
  \hat{V}_k+(\theta-\hat{V}_k)\kappa\Delta t+\sigma\sqrt{\hat{V}_k}\sqrt{\Delta t}Z^1_{k+1}
  .
\end{align}
For our method it is important that $\hat{V}$ is sampled solely from $Z^{1}$ and to switch to log-space.
This will be explained in the next sections.

Write $\bsZ=(Z^1_1,Z^2_1,Z^1_2,Z^2_2,\ldots,Z^2_m)'\in\mathbb{R}^{2m}$, where the prime is used to denote the transpose of a vector.
Then $\bsZ$ has multivariate standard normal distribution.
Assuming a European option payoff represented as
\begin{align*}
 &
  \max\left( f(\bsZ),0\right)
\end{align*}
one usually simulates the function $f(\bsZ)$ by mapping a uniform variate $\bsu$ in the unit cube to $\bsZ$ by applying the inverse cumulative distribution function $\Phi^{-1}$. 
We will call this method the standard Monte Carlo method (MC).
When using quasi-Monte Carlo (QMC), the uniform variates are replaced by a low-discrepancy point set.
Our conditional sampling scheme will use the influence of the first uniform variable $u_1$ to try and force the barrier conditions to be met.
For this we will employ the LT method.
First, the uniformly sampled variate $\bsu$ is mapped to a standard normal variate $\bsz$ as in the MC method.
The function $f(\bsZ)$ is then sampled using the transformation $\bsZ = Q \bsz$ for a carefully chosen orthogonal matrix $Q$.
%
%
%
This means that in\RefEq{eq:logS} and\RefEq{eq:V} we take, for $k=0,\ldots,m-1$,
\begin{align*}
  Z^1_{k+1} &= \sum_{n=1}^{2m} q_{2k+1,n} z_n
  &\text{and}&&
  Z^2_{k+1} &= \sum_{n=1}^{2m} q_{2k+2,n} z_n
  ,
\end{align*}
where $q_{i,j}$ denotes the element from the matrix $Q$ at row~$i$ and column~$j$.
We remark that, for ease of notation, we will write $f(\bsZ)$, $f(\bsz)$, $f(\bsu)$ or $f(\hat{S}_1, \ldots, \hat{S}_m)$ to denote the function $f$ from above in terms of normal variates $\bsZ$ or $\bsz$, uniform variates $\bsu$ or just the discretized stock path $\hat{S}_1, \ldots, \hat{S}_m$.

In what follows the notation $Q_{\bullet k}$ denotes the $k$th column of $Q$ and $Q_{k \bullet}$ denotes the $k$th row.
The LT method \cite{IT2006} chooses the matrix $Q$ according to the following optimization problem:
\begin{align*}
 \underset{Q_{\bullet k}\in\mathbb{R}^{2m}}{\text{maximize}} \qquad& \text{variance contribution of $f$ due to $k$th dimension} \\
  \text{subject to} \qquad& \lVert Q_{\bullet k} \rVert=1, \\
		    \qquad& \langle Q^*_{\bullet j},Q_{\bullet k} \rangle=0, \quad j=1,\ldots,k-1,
\end{align*}
where $Q^*_{\bullet j}$ denotes the columns of $Q$ that have already been optimized in the previous iterations.
The algorithm is carried out iteratively for $k=1,2,\ldots,2m$ so that in the $k$th optimization step the objective function ensures that, given columns $Q^*_{\bullet j}$, $j=1,\ldots,k-1$ which have already been determined in the previous iterations, the variance contribution due to the $k$th dimension is maximized while the constraints ensure orthogonality.
Being able to express the variance contribution for each component analytically for general payoff functions $f$ can be quite complicated.
Therefore, Imai and Tan \cite{IT2006} propose to approximate the objective function by linearizing it using a first-order Taylor expansion for $\bsz=\hat{\bsz}+\Delta \bsz$,
\begin{align*}
 f(\bsz)
  &\approx
  f(\hat{\bsz}) + \sum_{k=1}^{2m} \left.\frac{\partial f}{\partial z_{k}}\right|_{\bsz=\hat{\bsz}}\Delta z_{k}. 
\end{align*}
Using this expansion, the variance contributed due to the $k$th component is 
\begin{align*}
 \left( \left.\frac{\partial f}{\partial z_k}\right|_{\bsz=\hat{\bsz}}\right)^2.
\end{align*}
The expansion points are chosen as $\hat{\bsz}_k = (1,\ldots,1,0,\ldots,0)$, the vector with $k-1$ leading ones.
Different expansion points will lead to different transformation matrices; this particular choice allows for an efficient construction.
The optimization problem becomes
\begin{align}
 \underset{Q_{\bullet k}\in\mathbb{R}^{2m}}{\text{maximize}} \qquad& \left( \left.\frac{\partial f}{\partial z_k}\right|_{\bsz=\hat{\bsz}_k}\right)^2 \label{eq:optLT}\\
  \text{subject to} \qquad& \lVert Q_{\bullet k} \rVert=1, \nonumber\\
		    \qquad& \langle Q^*_{\bullet j},Q_{\bullet k} \rangle=0, \quad j=1,\ldots,k-1.\nonumber
\end{align}


The original Imai and Tan paper \cite{IT2006} considers a European call option to illustrate the computational advantage of the LT method under the Heston model.
In their paper the stochastic volatility is described in \cite[Sect.~4.2]{IT2006} and we will revisit their method in \RefSec{sect:ORLTCS}.
For ease of illustration we also consider the payoff function inside the $\max$-function to be that of a European call option
\begin{align*}
 f(\bsz)
  &=
  \hat{S}_m - K
\end{align*}
where $K$ is the strike price.
For notational ease, we introduce the following functions:
{ \allowdisplaybreaks
\begin{align*}
 f^1_k &= \frac{\sqrt{\Delta t}}{2\sqrt{\hat{V}_{k}}}\left(\rho Z_{k+1}^1+\sqrt{1-\rho^2}Z^2_{k+1}\right)-\frac{\Delta t}{2}, \\
 f^2_k &= \sqrt{\hat{V}_{k}}\sqrt{\Delta t}, \\
 f^3_k &= 1-\kappa\Delta t + \frac{\sigma\sqrt{\Delta t}}{2\sqrt{\hat{V}_k}} Z^1_{k+1}, \\
 f^4_k &= \sigma \sqrt{\hat{V}_{k}} \sqrt{\Delta t}. 
\end{align*}
}%
Note that all the above functions $f^i$ depend on $\bsZ$.
Similar to \cite{IT2006},
to find the partial derivatives $\partial \hat{S}_m/\partial z_i$ needed for the optimization algorithm, we obtain the recursive relations (with initial conditions $\partial \log \hat{S}_0 /\partial z_i=0$ and $\partial \hat{V}_0/\partial z_i=0$)
\begin{align}
  \frac{\partial \log\hat{S}_{k+1}}{\partial z_i}
  &=
  \frac{\partial\log\hat{S}_{k}}{\partial z_i} + \frac{\partial\hat{V}_k}{\partial z_i}f^1_{k} + \left(\rho q_{2k+1,i}+\sqrt{1-\rho^2}q_{2k+2,i}\right) f^2_{k}
  , \label{eq:recLT1} \\
  \frac{\partial \hat{V}_{k+1}}{\partial z_i}
 &=
 \frac{\partial \hat{V}_k}{\partial z_i} f^3_k + q_{2k+1,i} f^4_k
  \label{eq:recLT2}
  ,
\end{align}
where $k$ goes from $0$ to $m-1$.
The chain rule is used to obtain
\begin{align*}
 \frac{\partial \hat{S}_m}{\partial z_i}
 &=
 \hat{S}_m \frac{\partial \log \hat{S}_m}{\partial z_i}
 .
\end{align*}
We will use the following lemma to calculate the transformation matrix.
\begin{lemma}\label{lem:RecAid}
  The recursion
  \begin{align*}
    F_{k+1} &= a_k F_k + b_k q_k ,
    \\
    G_{k+1} &= c_k G_k + d_k q_k + e_k F_k,
  \end{align*}
  with initial values $F_0 = G_0 = 0$ can be written at index $k+1$ as a linear combination of the $q_{\ell}$, $\ell=0,\ldots,k$, as follows
  \begin{align*}
    F_{k+1} &= \sum_{\ell=0}^k q_{\ell} \, b_{\ell} \prod_{j=\ell+1}^k a_j ,
    \\
    G_{k+1} &= \sum_{\ell=0}^k q_{\ell}
      \left(
        d_{\ell} \prod_{j=\ell+1}^k c_j
        +
        b_{\ell} \sum_{t=\ell+1}^k e_t \prod_{v=t+1}^k c_v \prod_{v=\ell+1}^{t-1} a_v\right)
    .
  \end{align*}
\end{lemma}
\begin{proof}
  The formula for $F_{k+1}$ follows immediately by induction.
  For the expansion of $G_{k+1}$ we first rewrite this formula in a more explicit recursive form
  \begin{align*}
    G_{k+1}
    &=
    \sum_{\ell=0}^{k} q_{\ell} d_{\ell} \prod_{j=\ell+1}^k c_j
    +
    \sum_{\ell=0}^{k-1} q_{\ell} b_{\ell} \sum_{t=\ell+1}^k e_t \prod_{v=t+1}^k c_v  \prod_{v=\ell+1}^{t-1} a_v 
    \\
    &=
    \sum_{\ell=0}^{k} q_{\ell} d_{\ell} \prod_{j=\ell+1}^k c_j
    +
    \sum_{t=1}^k e_t \prod_{v=t+1}^k c_v \left( \sum_{\ell=0}^{t-1} q_{\ell} b_{\ell} \prod_{v=\ell+1}^{t-1} a_v \right)
    .
  \end{align*}
  The part in-between the braces equals $F_t$ and the proof now follows by induction on $k$.
\end{proof}
A similar result is obtained if the second recursion is replaced by $G_{k+1} = c_k G_k + d_k q_k + d'_k q'_k + e_k F_k$.
Furthermore the coefficients in the expansion for $q_{\ell}$ and $q'_{\ell}$ can cheaply be calculated recursively.
Using this lemma, we can make the log-LT construction for the Heston model explicit in the following lemma.
\begin{proposition}\label{prop:optvLTc}
 The column vector $Q_{\bullet k}$ that solves the optimization problem\RefEq{eq:optLT} for a call option under the Heston model is given by $Q_{\bullet k} = \pm \bsv / \|\bsv\|$ where
\begin{align*}
  v_{2\ell+1} &= \hat{S}_m f^2_{\ell} \rho + \hat{S}_m f^4_{\ell} \sum_{t=\ell+1}^{m-1}f^1_t\prod_{v=\ell+1}^{t-1}f^3_v
  ,
  \\
   v_{2\ell+2} &= \hat{S}_m f^2_{\ell} \sqrt{1-\rho^2},
\end{align*}
for $\ell=0,\ldots,m-1$.
\end{proposition}
\begin{proof}
  By \cite[Theorem 1]{IT2006} the solution to the optimization problem\RefEq{eq:optLT} is given by 
  \begin{align*}
   Q_{\bullet k}
   &=
   \pm \frac{\bsv}{\| \bsv \|}
   ,
  \end{align*} 
  where $\bsv$ is determined from
  \begin{align*}
    Q_{\bullet k}' \bsv
    &=
    \frac{\partial \hat{S}_m}{\partial z_k}
    =
    \hat{S}_m \frac{\partial \log \hat{S}_m}{\partial z_k}
    .
  \end{align*}
  With the help of \RefLem{lem:RecAid} we find from\RefEq{eq:recLT1} and\RefEq{eq:recLT2}
  \begin{align*}
    \frac{\partial \log \hat{S}_m}{\partial z_k}
    &=
    \sum_{\ell=0}^{m-1} q_{2\ell+1,k} \left( \rho f^2_{\ell} + f^4_{\ell} \sum_{t=\ell+1}^{m-1} f^1_t \prod_{v=\ell+1}^{t-1} f^3_v \right)
    +
    \sum_{\ell=0}^{m-1} q_{2\ell+2,k} \sqrt{1-\rho^2} f^2_{\ell}
    ,
  \end{align*}
  from which the result now follows. 
\end{proof}

Note that since $\hat{S}_m$ and all functions $f^i$ depend on $\bsZ$, the vector $\bsv$ changes in each iteration step of\RefEq{eq:optLT} as the reference point $\hat{\bsz}$ is changed.

This construction can also be used for a put option with payoff
\begin{align*}
 f(\bsz)
&=
K-\hat{S}_m.
\end{align*}
In case of an arithmetic Asian option, the payoff is given by
\begin{align*}
 f(\bsz)
  &=
  \frac{1}{m}\sum_{j=1}^m \hat{S}_j - K.
\end{align*}
In that case the optimization problem\RefEq{eq:optLT} contains the sum of partial derivatives
\begin{align*}
 \left.\frac{\partial f}{\partial z_k}\right|_{\bsz=\hat{\bsz}_k}
&=
 \frac{1}{m}\left.\sum_{j=1}^m \frac{\partial \hat{S}_j}{\partial z_k} \right|_{\bsz=\hat{\bsz}_k}
 .
\end{align*}
It is thus straightforward to use the results for the call option in \RefProp{prop:optvLTc} to construct the transformation matrix for the arithmetic Asian option.

Crucial to our conditional sampling algorithm is that we modify the LT construction by forcing all odd elements in the first column of $Q$ to zero, i.e., $q_{2k+1,1} = 0$ for $k=0,\ldots,m-1$.
This removes the influence of $z_1$ to $Z^1_k$ and thus $\hat{V}_k$ for all $k$.
The LT algorithm then finds the orthogonal matrix $Q$ which solves the optimization problem under this extra constraint (which fixes $m$ elements of the $4m^2$).
In the next section we will show this leads to an elegant conditional sampling scheme.
\begin{lemma}\label{lem:samesign}
  Under the condition that $q_{2\ell+1,1} = 0$ for $\ell=0,\ldots,m-1$ we have that the elements $q_{2\ell+2,1}$ all have the same sign.
\end{lemma}
\begin{proof}
  From \RefProp{prop:optvLTc}, for $k=1$, we find that $q_{2\ell+2,1}$ is proportional to $v_{2\ell+2}$, i.e.,
  \begin{align*}
    v_{2\ell+2}
    &=
    \hat{S}_m \sqrt{\hat{V}_{\ell}} \sqrt{\Delta t} \sqrt{1-\rho^2}
    ,
  \end{align*}
  which is always positive,
  and $q_{2\ell+1,1} = v_{2\ell+1} = 0$.
  Following \RefProp{prop:optvLTc} we now take $\pm \bsv / \|\bsv\|$ from which the result follows. 
\end{proof}


\section{Conditional sampling on log-LT}\label{sect:cs}

For expository reasons assume for now an up-\&-out option with barrier $B$,
\begin{align}
 g(\hat{S}_1,\ldots,\hat{S}_m)
 &=
 \max\left(f(\hat{S}_1,\ldots,\hat{S}_m),0\right) 
 \, \mathbb{I}\left\{\max_{k}\hat{S}_{k}<B\right\}
 .
 \label{eq:uao_opt}
\end{align} 
The condition at time $t_{k+1}$ that the asset stays below the barrier can then be written, for $k=0,\ldots,m-1$, as
\begin{align*}
 \log \hat{S}_{k+1}
 &=
 \log \hat{S}_k + \left(r-\frac{\hat{V}_k}{2}\right)\Delta t + \sqrt{\hat{V}_k}\sqrt{\Delta t}
  \left( \rho Z^1_{k+1} + \sqrt{1-\rho^2} Z^2_{k+1} \right) 
  \\
 &= \log S_0 + r(k+1)\Delta t - \Delta t\sum_{\ell=0}^k \frac{\hat{V}_{\ell}^2}{2}  \\
 &\qquad + \sum_{\ell=0}^k\sqrt{\hat{V}_{\ell}}\sqrt{\Delta t}\sum_{n=2}^{2m}\left( \rho q_{2\ell+1,n}+\sqrt{1-\rho^2} \, q_{2\ell+2,n}\right) z_n  \\
 &\qquad + z_1\sqrt{\Delta t}\sqrt{1-\rho^2}\sum_{\ell=0}^k\sqrt{\hat{V}_{\ell}} \,q_{2\ell+2,1}\\
 &< 
 \log B,
\end{align*}
where we have used $q_{2\ell+1,1} = 0$.
For notational ease we define the function
\begin{multline}
 \Gamma_k(B,\bsz_{2:2m})
  =
  \frac{\log B/S_0 - r(k+1)\Delta t + \Delta t\sum_{\ell=0}^k \hat{V}_\ell^2/2}
  {\sqrt{\Delta t}\sqrt{1-\rho^2}\sum_{\ell=0}^k\sqrt{\hat{V}_\ell} \, q_{2\ell+2,1}}
  \\
  -  
  \frac{\sum_{\ell=0}^k\sqrt{\hat{V}_\ell}\sqrt{\Delta t}\sum_{n=2}^{2m}\left( \rho q_{2\ell+1,n}+\sqrt{1-\rho^2} \, q_{2\ell+2,n}\right) z_n}
  {\sqrt{\Delta t}\sqrt{1-\rho^2}\sum_{\ell=0}^k\sqrt{\hat{V}_\ell} \, q_{2\ell+2,1}}
  .
\end{multline}
Here the notation $\bsz_{2:2m}$ is used to indicate the dependency on $z_2,\ldots,z_{2m}$, but not $z_1$.
Note that $\Gamma_k$ depends on all other market parameters as well, but this dependency is supressed not to clutter the formulas.
Because of the assumption that $q_{2k+1,1}=0$ for all $k$, $\hat{V}$ can be sampled independently of $z_1$.
This means the barrier condition can be written as a single condition on $z_1$, i.e.,
\begin{align*}
 z_1
 &< \min_k \Gamma_k(B,\bsz_{2:2m}) 
 \quad\text{if all $q_{2\ell+2,1} > 0$},
\end{align*}
and 
\begin{align*}
 z_1
 &> \max_k \Gamma_k(B,\bsz_{2:2m})  
 \quad\text{if all $q_{2\ell+2,1} < 0$}.
\end{align*}

The condition on $z_1$ was here derived for an up-\&-out option for ease of exposition.
The modifications for more complex barriers can easily be obtained from here.
\RefTab{table:bar_cons} gives an overview of the conditions on $z_1$ for the basic barrier types and shows that these conditions can easily be combined for more complex types.

\begin{table}[t]
\begin{center}
\hspace{-5mm}\begin{tabular}{ c | c }
\hline
Type & 
all $q_{2\ell+2,1} > 0$
\\
\hline 
U\&O ($B$) & $z_1<\min_k \Gamma_k(B,\bsz_{2:2m})$ \\
D\&O ($B$) & $z_1>\max_k \Gamma_k(B,\bsz_{2:2m})$ \\
U\&I ($B$) & $z_1>\min_k \Gamma_k(B,\bsz_{2:2m})$ \\
D\&I ($B$) & $z_1<\max_k \Gamma_k(B,\bsz_{2:2m})$ \\
U\&O + D\&O ($B_1>B_2$) & $z_1\in(\max_k \Gamma_k(B_2,\bsz_{2:2m}),\min_k \Gamma_k(B_1,\bsz_{2:2m}))$ \\
U\&O + D\&I ($B_1>B_2$) & $z_1< \min\{\max_k \Gamma_k(B_2,\bsz_{2:2m}),\min_k \Gamma_k(B_1,\bsz_{2:2m})\} $
\\
\hline
\hline
Type & 
all $q_{2\ell+2,1} < 0$
\\
\hline
U\&O ($B$) & $z_1>\max_k \Gamma_k(B,\bsz_{2:2m})$ \\
D\&O ($B$) & $z_1<\min_k \Gamma_k(B,\bsz_{2:2m})$ \\
U\&I ($B$) & $z_1<\max_k \Gamma_k(B,\bsz_{2:2m})$ \\
D\&I ($B$) & $z_1>\min_k \Gamma_k(B,\bsz_{2:2m})$ \\
U\&O + D\&O ($B_1>B_2$) & $z_1\in(\max_k \Gamma_k(B_1,\bsz_{2:2m}),\min_k \Gamma_k(B_2,\bsz_{2:2m}))$ \\
U\&O + D\&I ($B_1>B_2$) & $z_1>\max\{\max_k \Gamma_k(B_1,\bsz_{2:2m}),\min_k \Gamma_k(B_2,\bsz_{2:2m})\}$ 
\\ \hline
 \end{tabular}
\caption{\small{The barrier constraints on $z_1$ for different types of barriers: up-\&-out (U\&O), down-\&-out (D\&O), up-\&-in (U\&I), down-\&-in (D\&I) and some combinations.
}}
\label{table:bar_cons}
\end{center}
\end{table}

We now show the main results on our conditional sampling scheme.
Again, for expository reasons, specialized for the case of the up-\&-out option from above.
This result can easily be modified for other payout structures in the same spirit as the results in \RefTab{table:bar_cons}.
The following theorem holds for both the Monte Carlo method as for a randomly shifted quasi-Monte Carlo rule.
\begin{theorem}
  For the up-\&-out option\RefEq{eq:uao_opt} and assuming that we fixed $q_{2\ell+2,1} > 0$ for $\ell=0,\ldots,m-1$ \textup{(}see \RefLem{lem:samesign}\textup{)} the approximation based on sampling
  \begin{align*}
    \hat{g}(z_1, \ldots, z_m)
    &=
    \Phi\left( \min_{k} \Gamma_k(B, \bsz_{2:2m}) \right)
    \max\left(f(\hat{z}_1,z_2,\ldots,z_m),0\right)
  \end{align*}
  where, using the relation $z_1 = \Phi^{-1}(u_1)$,
  \begin{align}
    \label{eq:rescale}
    \hat{z}_1
    &=
    \Phi^{-1}\left( u_1 \min_{k} \Gamma_k(B, \bsz_{2:2m}) \right)
    ,
  \end{align}
  is unbiased.
  Furthermore, if we denote the respective unconditional method by
  \begin{align*}
    g(z_1, \ldots, z_m)
    &=
    \max\left(f(\hat{S}_1,\ldots,\hat{S}_m),0\right)
    \, \mathbb{I}\left\{\max_{k}\hat{S}_{k}<B\right\}
    ,
  \end{align*}
  where the $\hat{S}_1$, \ldots, $\hat{S}_m$ are obtained directly from $z_1$,\ldots,$z_m$ without using\RefEq{eq:rescale},
  then, when using the Monte Carlo method or a randomly shifted quasi-Monte Carlo method, the conditional sampling has reduced variance, i.e., $\Var[\hat{g}] \le \Var[g]$.
  Furthermore the inequality is strict if $\Prob[\max_{k} \hat{S}_k \geq B] >0$ and $\Exp[g] > 0$, i.e., if there is any chance of knock-out and positive payoff.
\end{theorem}
\begin{proof}
  The proof can be constructed similar to \cite[Theorem 3, 4 \& 5]{NicoBSLT} from our previous work. 
\end{proof}

The previous result shows that the proposed conditional algorithm can never do worse than its unconditional variant.
Furthermore, the more chance there is on a knock-out the more effect the conditional algorithm will have.
This can be observed in the examples in \RefSec{sec:LOGLTCSEX}.


\bigskip
\noindent\emph{Remark.}
The conditional sampling was applied to $z_1$ (or, equivalently, to $u_1$) to keep the asset from knocking out (or in). 
Taking it one step further one could try to add an additional bound on $z_1$, keeping $\bsz_{2:2m}$ constant, in order to force a strictly positive payout.
This is more involved than the barrier condition however as for more complicated payoffs than calls and puts there might not exist analytical formulae such as in \RefTab{table:bar_cons} to condition $z_1$.
It is interesting to note that for calls and puts the same formulas can be used as in \RefTab{table:bar_cons}, only now restricting the $\Gamma_k$ functions to $\Gamma_m(K, \bsz_{2:2m})$. 
Adding this constraint to the existing barrier conditions is straightforward.
Root finding methods can be employed for more complex payout structures.
See our previous work \cite{NicoBSLT} for a detailed analysis of root finding for Asian options.


\section{The original LT method for Heston}\label{sect:ORLTCS}

We mentioned previously that it is essential for our method to switch to log prices.
To illustrate the problem, we introduce the LT method for the Heston model as in~\cite{IT2006} and
we derive also an explicit form of the orthogonal matrix $Q$ (cf.~ \RefProp{prop:optvLTc}).
However, the conditional sampling scheme from the previous section is not applicable.
The Euler--Maruyama discretizations for $S(t)$ and $V(t)$ are given by
\begin{align*}
  \hat{S}_{k+1}
  &=
  \hat{S}_k + r \hat{S}_k \Delta t + \sqrt{\hat{V}_k} \hat{S}_k \sqrt{\Delta t} \left( \rho Z^1_{k+1} + \sqrt{1-\rho^2} Z^2_{k+1} \right)
  ,
  \\
  \hat{V}_{k+1}
  &=
  \hat{V}_k + (\theta-\hat{V}_k) \kappa \Delta t + \sigma \sqrt{\hat{V}_k}\sqrt{\Delta t} Z^1_{k+1}
  ,
\end{align*}
compare with\RefEqEq{eq:logS}{eq:V}.
For ease of notation, we introduce the following functions:
{ \allowdisplaybreaks
\begin{align*} 
 f^1_k &= 1+r\Delta t+\sqrt{\hat{V}_k}\sqrt{\Delta t}\left(\rho Z^1_{k+1} + \sqrt{1-\rho^2}Z^2_{k+1}\right), \\
 f^2_k &= \frac{\hat{S}_k\sqrt{\Delta t}}{2\sqrt{\hat{V}_k}} \left(\rho Z^1_{k+1} + \sqrt{1-\rho^2}Z^2_{k+1}\right), \\
 f^3_k &= \hat{S}_k\sqrt{\hat{V}_k}\sqrt{\Delta t}, \\
 f^4_k &= 1-\kappa\Delta t +\frac{\sigma\sqrt{\Delta t}}{2\sqrt{\hat{V}_k}}Z^1_{k+1}, \\
 f^5_k &= \sigma\sqrt{\hat{V}_k}\sqrt{\Delta t}.
\end{align*}
}%
Note that all the above functions $f^i$ depend on $\bsZ$.
The recursion relations for the partial derivatives become
\begin{align*}
 \frac{\partial \hat{S}_{k+1}}{\partial z_i}
 &=
 \frac{\partial \hat{S}_k}{\partial z_i} f^1_k +\frac{\partial \hat{V}_k}{\partial z_i} f^2_k+q_{2k+1,i}\rho f^3_k + q_{2k+2,i}\sqrt{1-\rho^2}f^3_k
 , 
 \\
 \frac{\partial \hat{V}_{k+1}}{\partial z_i}
 &=
 \frac{\partial \hat{V}_k}{\partial z_i} f^4_k + q_{2k+1,i} f^5_k   
 ,
\end{align*}
for $k=0,\ldots,m-1$, and initial conditions $\partial \hat{S}_0 / \partial z_i = 0$ and $\partial \hat{V}_0 / \partial z_i = 0$.
With this notation we obtain the LT construction for the Heston model in explicit form.
\begin{proposition}\label{prop:optvLTc2}
 The column vector $\bsv=Q_{\bullet k}$ that maximizes the optimization problem\RefEq{eq:optLT} for a call option under the Heston model is given by $Q_{\bullet k} = \pm \bsv / \|\bsv\|$ where
\begin{align*}
  v_{2\ell+1}
  &=
  f^3_{\ell} \rho \prod_{j=\ell+1}^{m-1} f^1_j + f^5_{\ell} \sum_{t=\ell+1}^{m-1} f^2_t \prod_{v=t+1}^{m-1} f^1_v \prod_{v=\ell+1}^{t-1} f^4_v
  ,
  \\
  v_{2\ell+2}
  &=
  f^3_{\ell} \sqrt{1-\rho^2} \prod_{j=\ell+1}^{m-1} f^1_j
  ,
\end{align*}
for $\ell=0,\ldots,m-1$.
\end{proposition}
\begin{proof}
 The proof is similar to \RefProp{prop:optvLTc}, again making use of \RefLem{lem:RecAid}.
\end{proof}

To show the advantage for conditional sampling of the log-LT method (as explained in \RefSecSec{sec:LTmethod}{sect:cs}) over this version we consider again the up-\&-out option with payoff
\begin{align*}
 g(\hat{S}_1,\ldots,\hat{S}_m)
 &=
 \max\left(f(\hat{S}_1,\ldots,\hat{S}_m),0\right)\,\mathbb{I}\left\{\max_{k}\hat{S}_{k}<B\right\}
 .
\end{align*}
The barrier condition at an arbitrary time step $t_{k+1}$ takes the following form:
\begin{align*}
 \hat{S}_{k+1}
 &=
 \hat{S}_k \left( 1 + r \Delta t + \sqrt{\hat{V}_k} \sqrt{\Delta t} \left( \rho Z^1_{k+1} + \sqrt{1-\rho^2} Z^2_{k+1} \right) \right)
 \\
 &=
 S_0 \prod_{\ell=0}^k \left( 1 + r \Delta t + \sqrt{\hat{V}_{\ell}} \sqrt{\Delta t} \sum_{n=1}^{2m} \left( \rho q_{2\ell+1,n} + \sqrt{1-\rho^2} q_{2\ell+2,n} \right) z_n \right)
 \\
 &<
 B.
\end{align*}
Trying to condition on $z_1$, as we did in the log-LT model (assuming again $q_{2\ell+1,1}=0$), leads to the following condition:
\begin{align*}
  \prod_{\ell=0}^k \left( A_{\ell} + \sqrt{\hat{V}_{\ell}} \sqrt{\Delta t} \sqrt{1-\rho^2} q_{2\ell+2,1} z_1 \right)
  < \frac{B}{S_0}
\end{align*}
where
\begin{align*}
  A_{\ell} &= 1 + r \Delta t + \sqrt{\hat{V}_{\ell}} \sqrt{\Delta t} \sum_{n=2}^{2m} \left( \rho q_{2\ell+1,n} + \sqrt{1-\rho^2} q_{2\ell+2,n} \right) z_n
  .
\end{align*}
%
To satisfy the condition on $z_1$, a $k+1$-th order polynomial must be solved in order to find the regions where the above condition holds.
To find the global condition, one has to solve polynomials of degrees $1$ to $m$, and then find the overlapping regions where all conditions hold.
This quickly becomes impractical and we therefore use the log-LT method which does not have this drawback.


\section{Examples}\label{sec:LOGLTCSEX}


\paragraph{\textit{Up-\&-out call and put}}
Consider the up-\&-out call and put options with payoffs 
\begin{align*} 
 P_c(\hat{S}_1,\ldots,\hat{S}_m)
 &=
 \max\left(\hat{S}_m-K,0\right)\,\mathbb{I}\left\{\max_{k}\hat{S}_{k}<B\right\}
 ,
 \\
 P_p(\hat{S}_1,\ldots,\hat{S}_m)
 &=
 \max\left(K-\hat{S}_m,0\right)\,\mathbb{I}\left\{\max_{k}\hat{S}_{k}<B\right\}
 .
\end{align*}
The fixed model parameters are $r=0\%$ and $\kappa=1$. Furthermore, time is discretized using $m=250$ steps and thus our stochastic dimension is~$250$.
The results for this example are calculated using a lattice sequence (with generating vector \textsf{\small exod8\_base2\_m13} from \cite{NUYWEB} constructed using the algorithm in \cite{CKN2006}). 
The improvements of the standard deviations w.r.t.\ the Monte Carlo method for different choices of $\rho$, $S_0$, $V_0=\theta=\sigma$, $K$ and $B$ are shown in \RefTab{table:UOC+P}.
The results for the call and put option seem to be consistent over all choices of parameters: the new conditional scheme (denoted by {\sc QMC+LT+CS}) improves significantly on the unconditional LT method (denoted by {\sc QMC+LT}).
Note that the {\sc QMC+LT} method uses the construction of \RefProp{prop:optvLTc2}.
Adding root finding (denoted by {\sc QMC+LT+CS+RF}), to force a positive payout, further dramatically improves the results.
The improvement of the {\sc QMC+LT+CS} method for the put option is even larger than that for the call option.
This difference should not come as a surprise: when using conditional sampling on a knock-out option, $z_1$ is modified such that the asset does not hit the barrier. 
In case of an up-\&-out call option, the asset paths are essentially pushed down in order to achieve this.
The payout of the call option however is an increasing function of $\hat{S}_m$, so that pushing the asset paths down has the side-effect of also pushing a lot of paths out of the money. 
For the put option the reverse is true: the payout is a decreasing function of $\hat{S}_m$, meaning that pushing the paths down will result in more paths ending up in the money. 
Root finding can be used to control this off-setting effect in case of the call option, this effect is clearly visible in \RefTab{table:UOC+P}.
These numerical results are illustrated in terms of $N$ in \RefFig{fig:CallConv} for two parameter choices for the call option.

\begin{table}[t]
\begin{center}\footnotesize
\hspace{-5mm}\begin{tabular}{ l B A c|c} 
$(V_0=\theta=\sigma, \rho, S_0, K, B)$ & QMC+LT+CS+RF & QMC+LT+CS & QMC+LT & Value
\\
\hline  
\multicolumn{5}{c}{\textbf{Call}} \\
\hline
$(0.2,-0.5,90,80,100)$ & $405$\% & $148$\% & $98$\% & $0.09$ \\  
$(0.2,-0.5,100,100,120)$ & $502$\% & $173$\% & $90$\% & $0.09$\\ 
$(0.2,-0.5,110,100,150)$ & $463$\% & $231$\% & $117$\% & $1.25$\\  
$(0.2,0.5,90,80,100)$ & $474$\% & $120$\% & $124$\% & $0.08$\\  
$(0.2,0.5,100,100,125)$ & $446$\% & $130$\% & $99$\% & $0.16$\\ 
$(0.2,0.5,110,100,140)$ & $454$\% & $166$\% & $136$\% & $0.56$\\ 
\hline 
$(0.3,-0.5,90,80,100)$ & $623$\% & $160$\% & $82$\% & $0.05$\\  
$(0.3,-0.5,100,100,120)$ & $590$\% & $160$\% & $144$\% & $0.06$\\  
$(0.3,-0.5,110,100,150)$ & $429$\% & $246$\% & $141$\% & $0.77$\\  
$(0.3,0.5,90,80,100)$ & $360$\% & $191$\% & $106$\% & $0.05$ \\ 
$(0.3,0.5,100,100,125)$ & $353$\% & $141$\% & $81$\% & $0.10$ \\ 
$(0.3,0.5,110,100,140)$ & $367$\% & $142$\% & $104$\% & $0.34$ \\
\hline  
\multicolumn{5}{c}{\textbf{Put}} \\
\hline
$(0.2,-0.5,90,80,100)$ & $367$\% & $331$\% & $184$\% & $9.02$\\ 
$(0.2,-0.5,100,100,105)$ & $279$\% & $235$\% & $126$\% & $7.76$\\  
$(0.2,-0.5,110,100,112)$ & $298$\% & $263$\% & $123$\% & $4.44$\\ 
$(0.2,0.5,90,80,100)$ & $361$\% & $376$\% & $148$\% & $6.05$\\
$(0.2,0.5,100,100,105)$ & $326$\% & $298$\% & $131$\% & $5.33$\\
$(0.2,0.5,110,100,112)$ & $317$\% & $325$\% & $149$\% & $2.98$ \\ 
\hline
$(0.3,-0.5,90,80,100)$ & $383$\% & $348$\% & $137$\% & $10.3$\\  
$(0.3,-0.5,100,100,105)$ & $260$\% & $243$\% & $144$\% & $8.65$ \\ 
$(0.3,-0.5,110,100,112)$ & $214$\% & $187$\% & $129$\% & $5.38$\\ 
$(0.3,0.5,90,80,100)$ & $380$\% & $294$\% & $160$\% & $6.44$ \\  
$(0.3,0.5,100,100,105)$ & $304$\% & $272$\% & $174$\% & $5.57$ \\ 
$(0.3,0.5,110,100,112)$ & $305$\% & $279$\% & $124$\% & $3.33$ \\
 \end{tabular}
 \caption{\small{Up-\&-out call and put. The reported numbers are the standard deviations of the MC method divided by those of the 
QMC+LT+CS+RF, QMC+LT+CS and QMC+LT methods. The MC method uses $30720$ samples, while the QMC methods use $1024$ samples and $30$ independent shifts. The rightmost column denotes the option value.
}}
\label{table:UOC+P}
\end{center}
\end{table}


\begin{figure}[t]
  \centering
    \hspace*{-7mm}\begin{tikzpicture}[scale=0.7]
      \begin{loglogaxis}[width=11cm,xlabel=$N$,ylabel=std.dev.,title={$(0.2,-0.5,110,100,150)$},log basis x=2]
	\addplot[mark=*, mark options={fill=white}] file {UOC1_CONV_MC.txt};
	\addplot[mark=triangle*, mark options={fill=white}] file {UOC1_CONV_LT.txt};
	\addplot[mark=square*, mark options={fill=white}] file {UOC1_CONV_CS.txt};
	\addplot[mark=square*, mark options={fill=black}] file {UOC1_CONV_RF.txt};
      \end{loglogaxis}
    \end{tikzpicture}%
    \begin{tikzpicture}[scale=0.7]
      \begin{loglogaxis}[width=11cm,xlabel=$N$,ylabel=std.dev.,title={$(0.3,0.5,100,100,120)$},log basis x=2]
	\addplot[mark=*, mark options={fill=white}] file {UOC2_CONV_MC.txt};
	\addplot[mark=triangle*, mark options={fill=white}] file {UOC2_CONV_LT.txt};
	\addplot[mark=square*, mark options={fill=white}] file {UOC2_CONV_CS.txt};
	\addplot[mark=square*, mark options={fill=black}] file {UOC2_CONV_RF.txt};
      \end{loglogaxis}
      \begin{scope}[shift={(.25,.25)}] 
	\draw (0,0) -- 
		plot[mark=square*, mark options={fill=black}] (0.25,0) -- (0.5,0) 
		node[right]{\tiny{QMC+LT+CS+RF}};
	\draw[yshift=\baselineskip] (0,0) -- 
		plot[mark=square*, mark options={fill=white}] (0.25,0) -- (0.5,0)
		node[right]{\tiny{QMC+LT+CS}};
	\draw[yshift=2\baselineskip] (0,0) -- 
		plot[mark=triangle*, mark options={fill=white}] (0.25,0) -- (0.5,0)
		node[right]{\tiny{QMC+LT}};
	\draw[yshift=3\baselineskip] (0,0) -- 
		plot[mark=*, mark options={fill=white}] (0.25,0) -- (0.5,0)
		node[right]{\tiny{MC}};
	\end{scope}
    \end{tikzpicture}
    \caption{\small{Up-\&-out call convergence plots for two options with different parameters. The fixed parameters are $r=0$\% and $\kappa=1$. 
    The different choices for $(V_0=\theta=\sigma,\rho,S_0,K,B)$ are denoted above the figures.}}
    \label{fig:CallConv}
\end{figure}
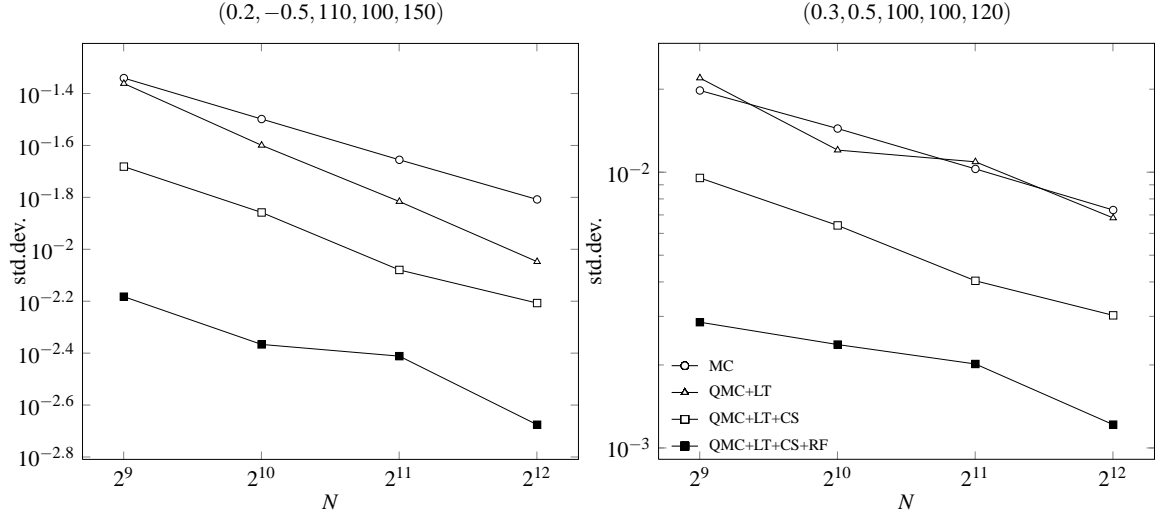

\paragraph{\textit{Up-\&-in call}}\label{sec:UIC+P}
Consider an up-\&-in call option with payoff
\begin{align*} 
 P(\hat{S}_1,\ldots,\hat{S}_m)
 &=
 \max\left(\hat{S}_m-K,0\right)\,\mathbb{I}\left\{\max_{k}\hat{S}_{k}>B\right\}
 .
\end{align*}
The fixed model parameters are $r=2\%$, $\kappa=1$ and $\sigma=0.2$. Again, $m=250$.
Here we use the Sobol' sequence with parameters from \cite{JK2008} and digital shifting \cite{DP2010}.
The standard deviations for different choices of $\rho$, $S_0$, $V_0=\theta$, $K$ and $B$ are shown in \RefTab{table:UIC}.
The improvements of the conditional scheme are extremely high for this case.
Note the impact of the correlation on the results: the improvement for $\rho=0.5$ is even approximately twice that for $\rho=-0.5$.
All parameter choices indicate that conditional sampling on the barrier condition greatly improves accuracy.
Adding the additional condition of the payout itself (root finding) provides another serious reduction in the standard deviation.

\begin{table}[t]
\begin{center}\footnotesize
\hspace{-5mm}\begin{tabular}{ l B A c|c} 
$(V_0=\theta, \rho, S_0, K, B)$ & QMC+LT+CS+RF &  QMC+LT+CS & QMC+LT & Value
\\
\hline  
$(0.1,-0.5,90,80,160)$ & $2158$\% & $1515$\% & $242$\% & $5.47$\\ 
$(0.1,-0.5,100,100,180)$ & $2377$\% & $1542$\% & $240$\% & $5.05$ \\ 
$(0.1,-0.5,110,120,200)$ & $2572$\% & $1545$\% & $250$\%  & $4.74$\\  
$(0.1,0.5,90,80,160)$ & $1557$\% & $654$\% & $341$\% & $17.4$\\ 
$(0.1,0.5,100,100,180)$ & $1564$\% & $644$\% & $354$\% & $16.9$\\ 
$(0.1,0.5,110,120,200)$ & $1556$\% & $640$\% & $373$\% & $16.6$\\ 
\hline
$(0.15,-0.5,90,80,160)$ & $2044$\% & $1247$\% & $366$\% & $10.6$ \\ 
$(0.15,-0.5,100,100,180)$ & $2243$\% & $1262$\% & $420$\% & $10.1$\\ 
$(0.15,-0.5,110,120,200)$ & $2391$\% & $1236$\% & $349$\% & $9.72$\\  
$(0.15,0.5,90,80,160)$ & $1570$\% & $568$\% & $421$\% & $23.3$\\  
$(0.15,0.5,100,100,180)$ & $1622$\% & $567$\% & $418$\% & $23.0$\\  
$(0.15,0.5,110,120,200)$ & $1649$\% & $562$\% & $366$\% & $22.9$
\end{tabular}
\caption\small{{Up-\&-in call. The reported numbers are the standard deviations of the MC method divided by those of the 
QMC+LT+CS+RF, QMC+LT+CS and QMC+LT methods. The MC method uses $30720$ samples, while the QMC+LT+CS+RF, QMC+LT+CS and QMC+LT methods use $1024$ samples and $30$ independent shifts. The rightmost column denotes the option value.
}}
\label{table:UIC}
\end{center}
\end{table}

\paragraph{\textit{Up-\&-out Asian}}\label{sec:UOA}
Consider an up-\&-out Asian option with payoff
\begin{align*} 
 P(\hat{S}_1,\ldots,\hat{S}_m)
 &=
 \max\left(\frac{1}{m}\sum_{k=1}^m\hat{S}_{k}-K,0\right)\,\mathbb{I}\left\{\max_{k}\hat{S}_{k}<B\right\}
 .
\end{align*}
The fixed model parameters are $r=5\%$, $\kappa=1$ and $\sigma=0.2$. The number of time steps is fixed at $m=250$.
We use the Sobol' sequence as in the previous example and the results are shown in \RefTab{table:UOAC}.
The results are once more very satisfactory with similar results as for the up-\&-out call and put options in \RefTab{table:UOC+P}.
\RefFig{fig:AsianConv} shows the convergence behaviour for two sets of parameter choices. 
As before, a significant variance reduction can be seen for our conditional sampling scheme and the root finding method further improves this result.

\begin{table}[t]
\begin{center}\footnotesize
\hspace{-5mm}\begin{tabular}{ l B A c|c} 
$(V_0=\theta, \rho, S_0, K, B)$ & QMC+LT+CS+RF &  QMC+LT+CS & QMC+LT & Value 
\\
\hline  
$(0.1,-0.5,90,80,120)$ & $483$\% & $329$\% & $154$\% & $1.70$\\
$(0.1,-0.5,100,100,140)$ & $461$\% & $245$\% & $185$\%  & $0.77$\\
$(0.1,-0.5,110,120,160)$ & $404$\% & $189$\% & $110$\% & $0.30$\\
$(0.1,0.5,90,80,120)$ & $392$\% & $328$\% & $144$\% & $1.34$\\
$(0.1,0.5,100,100,140)$ & $414$\% & $252$\% & $115$\% & $0.53$\\
$(0.1,0.5,110,120,160)$ & $502$\% & $209$\% & $133$\% & $0.18$ \\
\hline
$(0.15,-0.5,90,80,120)$ & $463$\% & $247$\% & $143$\%  & $0.77$\\
$(0.15,-0.5,100,100,140)$ & $425$\% & $183$\% & $125$\% & $0.29$\\
$(0.15,-0.5,110,120,160)$ & $389$\% & $161$\% & $93$\% & $0.10$\\
$(0.15,0.5,90,80,120)$ & $416$\% & $257$\% & $111$\% & $0.61$\\
$(0.15,0.5,100,100,140)$ & $486$\% & $201$\% & $119$\% & $0.20$\\
$(0.15,0.5,110,120,160)$ & $528$\% & $171$\% & $108$\% & $0.05$
\end{tabular}
\caption{\small{Up-\&-out Asian call. The reported numbers are the standard deviations of the MC method divided by those of the 
QMC+LT+CS+RF, QMC+LT+CS and QMC+LT methods. The MC method uses $30720$ samples, while the QMC+LT+CS+RF, QMC+LT+CS and QMC+LT methods use $1024$ samples and $30$ independent shifts. The rightmost column denotes the option value.
}}
\label{table:UOAC}
\end{center}
\end{table}

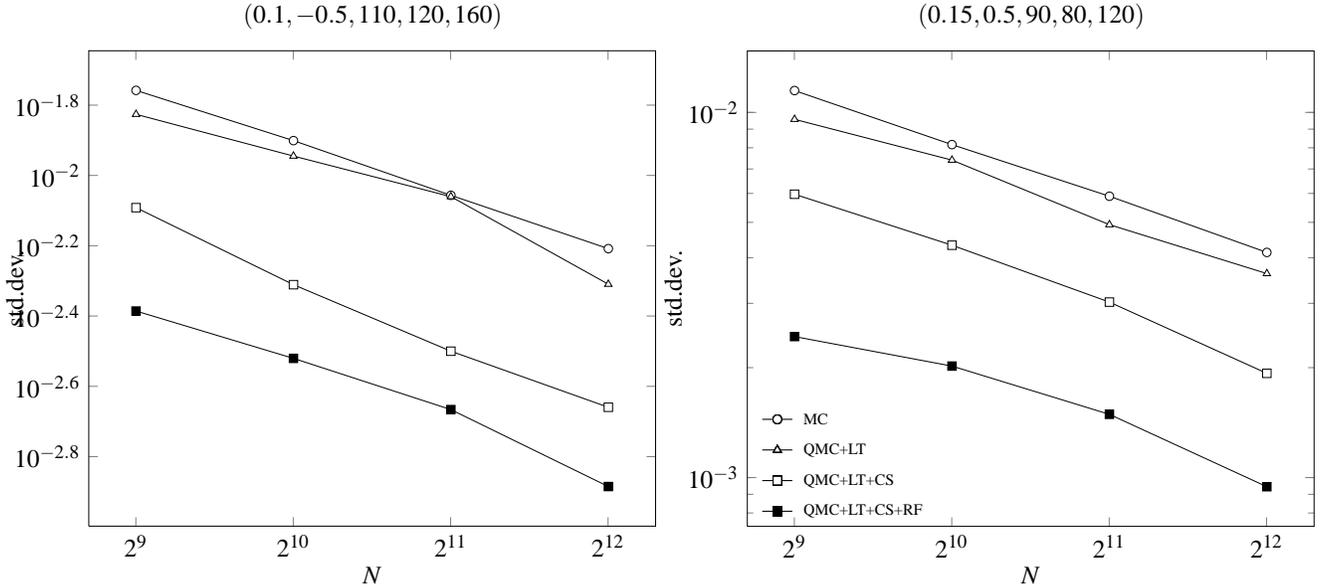
\begin{figure}[t]
  \centering
    \hspace*{-7mm}\begin{tikzpicture}[scale=0.8]
      \begin{loglogaxis}[width=11cm,xlabel=$N$,ylabel=std.dev.,title={$(0.1,-0.5,110,120,160)$},log basis x=2]
	\addplot[mark=*, mark options={fill=white}] file {UAC1_CONV_MC.txt};
	\addplot[mark=triangle*, mark options={fill=white}] file {UAC1_CONV_LT.txt};
	\addplot[mark=square*, mark options={fill=white}] file {UAC1_CONV_CS.txt};
	\addplot[mark=square*, mark options={fill=black}] file {UAC1_CONV_RF.txt};
      \end{loglogaxis}
    \end{tikzpicture}%
    \begin{tikzpicture}[scale=0.8]
      \begin{loglogaxis}[width=11cm,xlabel=$N$,ylabel=std.dev.,title={$(0.15,0.5,90,80,120)$},log basis x=2]
	\addplot[mark=*, mark options={fill=white}] file {UAC2_CONV_MC.txt};
	\addplot[mark=triangle*, mark options={fill=white}] file {UAC2_CONV_LT.txt};
	\addplot[mark=square*, mark options={fill=white}] file {UAC2_CONV_CS.txt};
	\addplot[mark=square*, mark options={fill=black}] file {UAC2_CONV_RF.txt};
      \end{loglogaxis}
      \begin{scope}[shift={(.25,.25)}] 
	\draw (0,0) -- 
		plot[mark=square*, mark options={fill=black}] (0.25,0) -- (0.5,0) 
		node[right]{\tiny{QMC+LT+CS+RF}};
	\draw[yshift=\baselineskip] (0,0) -- 
		plot[mark=square*, mark options={fill=white}] (0.25,0) -- (0.5,0)
		node[right]{\tiny{QMC+LT+CS}};
	\draw[yshift=2\baselineskip] (0,0) -- 
		plot[mark=triangle*, mark options={fill=white}] (0.25,0) -- (0.5,0)
		node[right]{\tiny{QMC+LT}};
	\draw[yshift=3\baselineskip] (0,0) -- 
		plot[mark=*, mark options={fill=white}] (0.25,0) -- (0.5,0)
		node[right]{\tiny{MC}};
	\end{scope}
    \end{tikzpicture}
    \caption{\small{Up-\&-out Asian call convergence plots for two options with different parameters. The fixed parameters are $r=5$\%, $\kappa=1$ and $\sigma=0.2$. 
    The different choices for $(V_0=\theta,\rho,S_0,K,B)$ are denoted above the figures.}}
    \label{fig:AsianConv}
\end{figure}

\section{Conclusion and outlook}

The conditional sampling scheme for the LT method introduced in \cite{NicoBSLT} for the Black--Scholes model has been extended to the Heston model. 
This was done by considering log prices and making the sampling of the volatility process independent of $z_1$.
We also obtained explicit constructions for the matrix $Q$ of the LT method.
The numerical results show that the method is very effective in reducing variance and outperforms the LT method by a huge margin.
We only considered an Euler--Maruyama discretization scheme for the asset and volatility processes.
It might be interesting to see if the theory and results carry over when other simulation methods are used, see \cite{VHP2010} for an overview of other methods.



\clearpage
\begin{thebibliography}{99.}%

\bibitem{NicoBSLT}
N.~Achtsis, R.~Cools and D.~Nuyens.
\newblock Conditional sampling for barrier option pricing under the {LT}
  method.
\newblock Submitted. Available at \url{http://arxiv.org/abs/1111.4808}

\bibitem{CKN2006}
R.~Cools, F.~Y. Kuo, and D.~Nuyens.
\newblock Constructing embedded lattice rules for multivariate integration.
\newblock {\em SIAM Journal on Scientific Computing}, 28(6):2162--2188, 2006.

\bibitem{DP2010}
J.~Dick and F.~Pillichshammer.
\newblock {\em Digital Nets and Sequences: Discrepancy Theory and Quasi-Monte
  Carlo Integration}.
\newblock Cambridge University Press, 2010.

\bibitem{GKSW2008}
M.~B. Giles, F.~Y. Kuo, I.~H. Sloan, and B.~J. Waterhouse.
\newblock Quasi-{M}onte {C}arlo for finance applications.
\newblock {\em ANZIAM Journal}, 50:308--323, 2008.

\bibitem{Glass2003}
P.~{Glasserman}.
\newblock {\em {M}onte {C}arlo Methods in Financial Engineering}.
\newblock Springer, 2003.

\bibitem{GS2001}
P.~{Glasserman} and J.~{Staum}.
\newblock Conditioning on one-step survival for barrier option simulations.
\newblock {\em Operations Research}, 49(6):923--937, 2001.

\bibitem{JK2008}
S.~{Joe} and F.~Y. {Kuo}.
\newblock Constructing {S}obol' sequences with better two-dimensional
  projections.
\newblock {\em SIAM Journal of Scientific Computing}, 30:2635--2654, 2008.

\bibitem{Heston1993}
S.L.\ {Heston}.
\newblock A closed-form solution for options with stochastic volatility with applications to bond and currency options.
\newblock {\em The Review of Financial Studies}, 6(2):327--343, 1993. 

\bibitem{IT2006}
J.~{Imai} and K.~S.\ {Tan}.
\newblock A general dimension reduction technique for derivative pricing.
\newblock {\em Journal of Computational Finance}, 10(2):129--155, 2006.

\bibitem{KP1992}
P.E.\ Kloeden and E.~Platen.
\newblock Numerical Solution of Stochastic Differential Equations. 
\newblock Springer-Verlag, 1992.

\bibitem{LEC2009}
 P.~L'{\'E}cuyer.
\newblock Quasi-Monte Carlo methods with applications in finance.
\newblock {\em Finance and Stochastics}, 13(3):307--349, 2009.

\bibitem{NW2012}
D.~Nuyens and B.J.~Waterhouse.
\newblock A global adaptive quasi-{M}onte {C}arlo algorithm for functions of
  low truncation dimension applied to problems from finance.
\newblock In H.~Wo{\'z}niakowski and L.~Plaskota, editors, {\em {M}onte {C}arlo
  and Quasi-{M}onte {C}arlo Methods 2010}, pages 589--607. Springer-Verlag, 2012.

\bibitem{NUYWEB}
\url{http://people.cs.kuleuven.be/~dirk.nuyens/qmc-generators} (27/07/2012)

\bibitem{VHP2010}
A.~{Van Haastrecht} and A.A.J.~{Pelsser}.
\newblock Efficient, almost exact simulation of the Heston stochastic volatility model.
\newblock {\em International Journal of Theoretical and Applied Finance}, 31(1):1--43, 2010.

\end{thebibliography}
\end{document}